\documentclass[journal]{IEEEtran}

\usepackage{cite}
\usepackage[colorlinks,citecolor=blue]{hyperref}
\usepackage{amsmath,amssymb,amsfonts}
\usepackage{graphicx}
\usepackage{textcomp}
\usepackage{mathtools}
\usepackage{mathbbol}
\usepackage{tikz}
\usepackage{verbatim}
\usepackage{moreverb}
\usepackage{graphics} 
\usepackage{graphicx}
\usepackage{subcaption}
\usepackage{times}
\usepackage{color}
\usepackage{lineno}
\usepackage{kotex}
\usepackage{caption}
\usepackage{booktabs,siunitx}
\usepackage{algorithm} 
\usepackage{algorithmicx} 
\usepackage{algpseudocode}
\usepackage{array}

\newtheorem{definition}{Definition}[section]

\newtheorem{remark}{Remark}[section]
\newtheorem{lemma}{Lemma}[section]
\newtheorem{theorem}{Theorem}[section]

\newenvironment{proof}{\par\noindent{\itshape Proof.}}{\hfill$\square$\par\medskip}

\begin{document}

\title{Topological Clusters in Multi-Agent Networks: Analysis and Algorithm}

\author{Jeong-Min Ma, Hyung-Gon Lee, Kevin L. Moore, Hyo-Sung Ahn, and Kwang-Kyo Oh
\thanks{J.-M. Ma, H.-S. Ahn and H.-G. Lee are with School of Mechanical Engineering, Gwangju Institute of Science and Technology, Gwangju, Republic of Korea(e-mail: majm1128@gist.ac.kr, hyosung@gist.ac.kr), hyunggohnlee@gist.ac.kr. }
\thanks{K.-K. Oh is with Department of Electrical Engineering, Sunchon National University, Sunchon, Republic of Korea (e-mail: kkoh@scnu.ac.kr).}
\thanks{K. L. Moore is with Department of Electrical Engineering, Colorado School of Mines, Golden, CO, USA (e-mail: kmoore@mines.edu).}}
\maketitle

\begin{abstract}
We study clustering properties of networks of single-integrator nodes over a directed graph, in which the nodes converge to steady-state values. These values define clustering groups of nodes, which depend on interaction topology, edge weights, and initial values. Focusing on the interaction topology of the network, we introduce the notion of topological clusters, which are sets of nodes that converge to an identical value due to the topological characteristics of the network, independent of the value of the edge weights. We then investigate properties of topological clusters and present a necessary and sufficient condition for a set of nodes to form a topological cluster. We also provide an algorithm for finding topological clusters. Examples show the validity of the analysis and algorithm.
\end{abstract}


\section{Introduction}
Multi-agent systems have attracted a significant amount of research interest due to their theoretical challenges and richness of applications \cite{olfatisaber2007consensus,ren2007information,oh2015survey}. In particular, directed networks have received considerable attention \cite{8444704,9866554}. In directed networks, agents often tend to aggregate into clusters depending on factors such as the values of the network's edge weights, node dynamics, and the topology of network.
There have been many studies of clustering behavior \cite{yu2010group,yu2012group,6760299,xia2011clustering,yi2011reaching,MONACO201953}. In \cite{yu2010group,yu2012group,6760299}, group consensus has been investigated in directed networks with time-varying topology, communication delays, or negative couplings. The authors of \cite{yu2010group,yu2012group} have provided algebraic conditions for clustering by dividing a network into several sub-networks. In \cite{6760299}, clustering behavior has been investigated for a network with some negative couplings. The authors of \cite{xia2011clustering} have studied clustering in networks with heterogeneous node dynamics, communication delay, or negative couplings. The authors of \cite{yi2011reaching} have studied the relationship between clustering and network topology. They have shown how the number of clusters can be determined by the Laplacian of a directed graph. The authors of \cite{MONACO201953} have studied clustering in networks with equal edge weights, showing that clusters correspond to cells of what is called the almost equitable graph partition. 

Identification of clusters has been studied in \cite{MONACO201953} for a network of single-integrator nodes. The authors of \cite{MONACO201953} have presented a method to identify clusters of a single-integrator network with equal edge weights. The method in \cite{MONACO201953} can be extended to a network with positive edge weights. However, since the method relies on algebraic properties of the graph Laplacian, it is difficult to apply when edge weights are not exactly defined or change frequently. In other words, the analysis in \cite{MONACO201953} is dependent on edge weights. Here, we study purely topological characteristics of clustering behavior, independent of edge weights of networks. 
Specifically, we introduce the notion of topological clusters defined independently of edge weights for a network of single-integrator nodes operating over a directed graph. Then, we present an algorithm to identify topological clusters of the network. 

We comment that the notion of topological clusters is useful for analyzing and controlling networks when edge weights are not exactly defined or change frequently. Further, knowledge of topological clusters can become an adequate theoretical basis for various control problems related to cluster behavior \cite{CACACE2021109734,8908749,9424396}. For example, the topological clusters associated with thermal dynamics of a building can be the basis for identifying suitable input-output pairing of sensors and actuators in a multi-zone control scheme \cite{6387350}. Also, knowledge of topological clusters in a social network can be used to analyze the behavior of opinion dynamics and the propagation of opinions \cite{090766188,Mas2010Individualization,dietrich2017control}. In addition, the notion of topological clusters can be applied to swarm control of UAVs \cite{s21113820}.

Consequently, the contributions of this paper can be summarized as follows:
First, we introduce the notion of topological clusters and present a necessary and sufficient condition for a set of nodes to form a topological cluster. Our notion of topological clusters is independent of edge weights, as opposed to \cite{yi2011reaching,MONACO201953}. That is, we always obtain the same topological clusters for any positive edge weights. Second, we propose an algorithm to identify topological clusters from the interaction topology of the network. The algorithm can be applied to social networks such as DeGroot model \cite{degroot1974reaching}. A preliminary, and partial, version of this algorithm is found in \cite{ma2020clusters}. Finally, we validate the proposed algorithm based on a real world example. This shows that the proposed algorithm can be useful for analyzing real world networks.

The outline of this paper is as follows. 
Following preliminaries in the next section, topological clusters are defined and their properties are investigated in Section~\ref{sec_clusters}. In Section~\ref{sec_clustering_algorithm}, an algorithm for finding topological clusters is proposed. In Section~\ref{sec_examples}, the proposed algorithm is applied to some examples, including a social network, and compared with the existing works. Concluding remarks are then provided in Section~\ref{sec_conclusion}.

\section{Preliminaries} \label{sec_preliminaries}
The set of real numbers is denoted by $\mathbb{R}$. A zero matrix is denoted by $\mathbf{0}$. A $n$-dimensional vector with ones is denoted by $\mathbf{1}_{n}$. The cardinality of a set $\mathcal{A}$ is denoted by $|\mathcal{A}|$. Given two sets, $\mathcal{A}$ and $\mathcal{B}$, 
$\mathcal{A} \backslash \mathcal{B}$ denotes the set of elements that are in $\mathcal{A}$ but not in $\mathcal{B}$.


A directed graph (or graph) $\mathcal{G} = (\mathcal{V}, \mathcal{E}, A)$ is defined as a triple consisting of a node set $\mathcal{V} = \{ 1, \ldots, N \}$, an edge set $\mathcal{E} \subset \{ (i, j): i, j \in \mathcal{V}, i \not =  j \}$, where an edge is an ordered pair of distinct nodes in $\mathcal{V}$, and a weight matrix $A = [a_{ij}] \in \mathbb{R}^{N \times N}$, which is a nonnegative matrix such that $a_{ij} > 0$ if and only if $(j,i) \in \mathcal{E}$. The weight matrix is also called the adjacency matrix of $\mathcal{G}$. If $a_{ij} = a_{ji}$, $\mathcal{G}$ is called undirected. If $(i, j)$ is an edge of $\mathcal{G}$, $i$ is a parent node of $j$ and $j$ is a child node of $i$. The set of neighbors of node $i$ in $\mathcal{G}$ is defined as $\mathcal{N}_{i} = \{ j: (j,i) \in \mathcal{E} \}$. 

A subgraph $\mathcal{G}_{s} = (\mathcal{V}_{s}, \mathcal{E}_{s})$ of $\mathcal{G}$ is a graph such that $\mathcal{V}_{s} \subseteq \mathcal{V}$ and $\mathcal{E}_{s} \subseteq \mathcal{E}$. If $\mathcal{V}_{s} = \mathcal{V}$, $\mathcal{G}_{s}$ is called a spanning subgraph of $\mathcal{G}$. 
Similarly, the induced subgraph is a graph whose vertex set is $\mathcal{V}_{s}$ and whose edge set consists of all of the edges in $\mathcal{E}$ that have both endpoints in $\mathcal{V}_{s}$. It is also said that $\mathcal{G}_{s}$ is induced by $\mathcal{V}_{s}$. The subgraph induced by $\mathcal{V}_{s}$ is denoted by $\mathcal{G}[\mathcal{V}_{s}]$. A tree is a graph such that every node except for a node called the root has exactly one parent. A forest is a graph consisting of one or more trees, no two of which have a node in common. A spanning tree (spanning forest) of $\mathcal{G}$ is a tree (forest) that is a spanning subgraph of $\mathcal{G}$. A path in $\mathcal{G}$ is a sequence $i_{1},\ldots,i_{k}$ of nodes such that $(i_{s}, i_{s-1}) \in \mathcal{E}$ for $s = 1,\ldots,k-1$. If there is a path from $i$ to $j$, we say $j$ can be reached from $i$. A weak path in $\mathcal{G}$ is a sequence $i_{1},\ldots, i_{k}$ of nodes such that either $(i_{s}, i_{s-1}) \in \mathcal{E}$ or $(i_{s-1}, i_{s}) \in \mathcal{E}$  for $s = 1,\ldots,k-1$. A graph $\mathcal{G}$ is strongly connected (weakly connected) if there exists a path (weak path) from $i$ to $j$ for any distinct nodes $i,j \in \mathcal{G}$. 



The Laplacian matrix $L = [l_{ij}] \in \mathbb{R}^{N \times N}$ of $\mathcal{G}$ is defined as
\begin{align}
l_{ij} = \left \{ \begin{array}{ll}
\sum_{j \in \mathcal{N}_{i}} a_{ij}, & i=j, \\
-a_{ij},  & (j,i) \in \mathcal{E},  \\
0, & \mathrm{otherwise}.
\end{array}	\right.	\nonumber
\end{align}
The following lemmas are useful for analyzing consensus networks:
\begin{lemma} [\cite{godsil2001algebraic}] \label{lemma_Laplacian_1} 
For the Laplacian matrix of a graph, the following is true:
\begin{itemize}
    \item All the eigenvalues of the Laplacian matrix have nonnegative real parts;
    \item Zero is an eigenvalue of the Laplacian matrix with $\mathbf{1}_{N}$ as the corresponding right eigenvector.
\end{itemize} 
\end{lemma}
\begin{lemma} [\cite{lin2005necessary,ren2005consensus}] \label{lemma_Laplacian_2}
For the Laplacian matrix of a graph, zero is a simple eigenvalue if and only if the graph has a directed spanning tree.
\end{lemma}
\begin{lemma} [\cite{niezink2011consensus}] \label{lemma_partitioning_G}
If a graph $\mathcal{G} = (\mathcal{V},\mathcal{E}, A)$ has no directed spanning tree, its nodes can be partitioned into $\mathcal{A}$, $\mathcal{B}$, and $\mathcal{C}$ such that
\begin{itemize}
    \item The node sets $\mathcal{A}$ and $\mathcal{B}$ are nonempty and $\mathcal{G}[\mathcal{A}]$ is strongly connected;
    \item No nodes in $\mathcal{A}$ and $\mathcal{B}$ have incoming edges from $\mathcal{V}\backslash\mathcal{A}$ and $\mathcal{V}\backslash\mathcal{B}$, respectively;
    \item The node set $\mathcal{C}$ corresponds to $\mathcal{V}\backslash(\mathcal{A} \cup \mathcal{B})$.
\end{itemize}
\end{lemma}
The key idea of \textbf{Lemma \ref{lemma_partitioning_G}} can be summarized as follows: Let $\mathcal{G}_{c} = (\mathcal{V}_{c},\mathcal{E}_{c})$ be the condensation of $\mathcal{G} = (\mathcal{V},\mathcal{E})$, which is a acyclic graph formed by contracting strongly connected components of $\mathcal{G}$ \cite{bullo2020lectures}. Since $\mathcal{G}$ has no directed spanning tree, $\mathcal{G}_{c}$ is acyclic; thus there is $v_{\mathcal{A}} \in \mathcal{V}_{c}$ having no parent nodes. The set $\mathcal{A} \subset \mathcal{V}$ corresponds to $v_{\mathcal{A}}$. Further $\mathcal{C} \subset \mathcal{V}$ corresponds to the nodes that are reachable from $v_{\mathcal{A}}$. Finally, the set $\mathcal{B}=\mathcal{V}\backslash(\mathcal{A} \cup \mathcal{C})$ is nonempty because $\mathcal{G}_{c}$ has no directed spanning tree. Then, no nodes in $\mathcal{A}$ and $\mathcal{B}$ have incoming edges from $\mathcal{V}\backslash\mathcal{A}$ and $\mathcal{V}\backslash\mathcal{B}$, respectively. It follows from \textbf{Lemma \ref{lemma_partitioning_G}} that if $\mathcal{G}$ has no directed spanning tree, its Laplacian $L$ can be partitioned as follows:
\begin{align*}
    L = \left[ \begin{matrix}
	L_{\mathcal{A}}&		\mathbf{0}&		\mathbf{0}\\
	\mathbf{0}&		L_{\mathcal{B}}&		\mathbf{0}\\
L_{\mathcal{CA}}&		L_{\mathcal{CB}}&	L_{\mathcal{C}}\\
\end{matrix} \right].
\end{align*}
This shows that $L_{\mathcal{A}}$ and $L_{\mathcal{B}}$ correspond to the Laplacian matrix of induced subgraph $\mathcal{G}[\mathcal{A}]$ and $\mathcal{G}[\mathcal{B}]$, respectively. On the other hand, $L_{\mathcal{C}}$ is not the Laplacian matrix of $\mathcal{G}[\mathcal{C}]$ because either $L_{\mathcal{CA}}$ or $L_{\mathcal{CB}}$ is nonzero.

In this paper, we study the following diffusively coupled dynamics with positive edge weights:
\begin{align}	\label{consensus_protocol_1}
\dot{x}_{i} =\sum_{j \in \mathcal{N}_{i}} a_{ij} (x_{j} - x_{i}),
\end{align}
where ${x}_{i}$ is a state of node $i$ and ${a}_{ij}$ is a weight of $(j,i)$. Without losing generality, assume that $n = 1$. The node dynamics in \eqref{consensus_protocol_1} can be written as
\begin{align}	\label{consensus_protocol_2}
\dot{x} = -Lx,
\end{align}
where $x = [x_{1} \cdots x_{N}]^{T}$ and $L$ is the Laplacian matrix of $\mathcal{G}$. 
It is said that the network \eqref{consensus_protocol_2} over $\mathcal{G}$ reaches consensus if $x_{i} - x_{j} \to 0$ as $t \to \infty$ for any nodes $i,j$ of $\mathcal{G}$. It is well-known that the coupled network \eqref{consensus_protocol_2} over a graph $\mathcal{G}$ reaches consensus if and only if $\mathcal{G}$ has a directed spanning tree.

If $\mathcal{G}$ has no directed spanning tree, the network \eqref{consensus_protocol_2} cannot reach consensus. To discuss this, we need the concept of leader and followers.
%
For a graph, a node is a leader node if it has no neighbor, and a follower node if it has at least one neighbor.
Denote the set of leader nodes by $\mathcal{L}$ and the set of follower nodes by $\mathcal{F}$. The consensus protocol in \eqref{consensus_protocol_1} can be written as follows \cite{cao2009containment}:
\begin{subequations} \label{consensus_protocol_3}	
\begin{align} 
\dot{x_{i}}  &= \sum_{j \in \mathcal{N}_{i}} a_{ij} (x_{j} - x_{i}),	\; i \in \mathcal{F}, \\
\dot{x_{i}}  &= 0, \; i \in \mathcal{L}.	
\end{align}
\end{subequations}
All the follower nodes of \eqref{consensus_protocol_3} over a graph $\mathcal{G}$ converge to the stationary convex hull spanned by the stationary leader nodes if and only if $\mathcal{G}$ has a directed spanning forest \cite{cao2009containment,liu2012necessary}. 
\section{Clusters in Consensus Networks}    \label{sec_clusters}



In this section, we first review clustering behavior in single-integrator node networks, which is dependent on edge weights. Then, we develop the notion of topological clusters, which is independent of edge weights. Some properties of topological clusters are also investigated.

\subsection{Clusters of networks over weighted graphs}


\begin{figure}
\centering
\begin{subfigure}[b]{0.45\linewidth}
\centering
\begin{tikzpicture}
\draw[thick] (-1.2,1) circle (0.21cm);          
\node at (-1.2,1) {1};
\draw[thick] (0,1) circle (0.21cm);             
\node at (0,1) {2};
\draw[thick] (-0.6,0) circle (0.21cm);          
\node at (-0.6,0) {3};
\draw[thick] (1.2,1) circle (0.21cm);           
\node at (1.2,1) {4};
\draw[thick] (2,0.3) circle (0.21cm);           
\node at (2,0.3) {5};
\draw[thick] (1.2,-0.8) circle (0.21cm);        
\node at (1.2,-0.8) {6};
\draw[thick] (-0.1,-1.3) circle (0.21cm);       
\node at (-0.1,-1.3) {7};
\draw[thick,<->] (-0.99,1) -- (-0.21,1);        
\draw[thick,->] (-1.1,0.83) -- (-0.7,0.17);     
\draw[thick,<->] (1.36,0.87) -- (1.84,0.43);    
\draw[thick,->] (-0.42,-0.1) -- (1.02,-0.7);    
\draw[thick,->] (1.2,0.79) -- (1.2,-0.59);      
\draw[thick,->] (-0.56,-0.2) -- (-0.14,-1.1);   
\draw[thick,->] (1.1,0.82) -- (0,-1.12);        
\draw[gray,thick,dashed] (-0.6,0.66) circle (1cm);
\node at (0,0.33) {$\mathcal{H}_1$};
\draw[rotate=-45,gray,thick,dashed] (0.63,1.6) ellipse (0.9cm and 0.5cm);
\node at (1.8,1.5) {$\mathcal{H}_2$};
\draw[rotate=22,gray,thick,dashed] (0.12,-1.15) ellipse (1.1cm and 0.5cm);
\node at (0.55,-1.3) {$\mathcal{C}$};
\node at (-0.6,1.2) {1};
\node at (-1.05,0.4) {1};
\node at (1.7,0.8) {1};
\node at (-0.55,-0.75) {1};
\node at (1.35,-0.1) {1};
\node at (0.8,-0.1) {1};
\node at (0.7,-0.8) {1};
\end{tikzpicture}
\caption{3 clusters}
\label{fig_cluster_ex_1}
\end{subfigure}
\hspace{3mm}
\begin{subfigure}[b]{0.45\linewidth}
\centering
\begin{tikzpicture}
\draw[thick] (-1.2,1) circle (0.21cm);          
\node at (-1.2,1) {1};
\draw[thick] (0,1) circle (0.21cm);             
\node at (0,1) {2};
\draw[thick] (-0.6,0) circle (0.21cm);          
\node at (-0.6,0) {3};
\draw[thick] (1.2,1) circle (0.21cm);           
\node at (1.2,1) {4};
\draw[thick] (2,0.3) circle (0.21cm);           
\node at (2,0.3) {5};
\draw[thick] (1.2,-0.8) circle (0.21cm);        
\node at (1.2,-0.8) {6};
\draw[thick] (-0.1,-1.3) circle (0.21cm);       
\node at (-0.1,-1.3) {7};
\draw[thick,<->] (-0.99,1) -- (-0.21,1);        
\draw[thick,->] (-1.1,0.83) -- (-0.7,0.17);     
\draw[thick,<->] (1.36,0.87) -- (1.84,0.43);    
\draw[thick,->] (-0.42,-0.1) -- (1.02,-0.7);    
\draw[thick,->] (1.2,0.79) -- (1.2,-0.59);      
\draw[thick,->] (-0.56,-0.2) -- (-0.14,-1.1);   
\draw[thick,->] (1.1,0.82) -- (0,-1.12);        
\draw[gray,thick,dashed] (-0.6,0.66) circle (1cm);
\node at (0,0.33) {$\mathcal{H}_1$};
\draw[rotate=-45,gray,thick,dashed] (0.63,1.6) ellipse (0.9cm and 0.5cm);
\node at (1.8,1.5) {$\mathcal{H}_2$};
\draw[gray,thick,dashed] (1.2,-0.8) circle (0.35cm);
\node at (0.5,-1.3) {$\mathcal{C}_1$};
\draw[gray,thick,dashed] (-0.1,-1.3) circle (0.35cm);
\node at (1.75,-1.1) {$\mathcal{C}_2$};
\node at (-0.6,1.2) {1};
\node at (-1.05,0.4) {2};
\node at (1.7,0.8) {1};
\node at (-0.55,-0.75) {3};
\node at (1.35,-0.1) {4};
\node at (0.8,-0.1) {1};
\node at (0.65,-0.8) {2};
\end{tikzpicture}
\caption{4 clusters}
\label{fig_cluster_ex_2}
\end{subfigure}
\caption{Clusters of networks with identical graph topology but different edge weights}
\label{fig_cluster}
\end{figure}
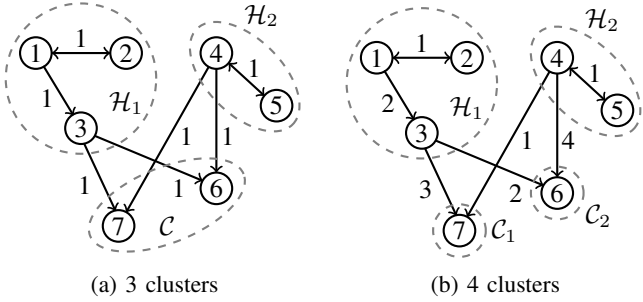

Consider a network of single-integrator nodes over a graph. Though the network cannot reach a consensus in the absence of a spanning tree, it can be shown that the state of each node converges to a steady-state value. Further, some nodes can converge to the same value to form a node group, which we call a cluster. 

For a network over a weighted graph, the authors of \cite{MONACO201953} have proposed a method to identify clusters based on algebraic properties of the Laplacian matrix. According to \cite{MONACO201953}, the Laplacian matrix of a weighted graph can be written as
\begin{equation} \label{transformed laplacian}
\bar{L}=\left[ \begin{matrix}
	L_1&		0&		\cdots&		0&		0\\
	0&		L_2&		\cdots&		0&		0\\
	\vdots&		\vdots&		\ddots&		\vdots&		\vdots\\
	0&		0&		\cdots&		L_{\mu}&		0\\
	M_1&		M_2&		\cdots&		M_{\mu}&		M\\
\end{matrix} \right] 
\end{equation}
based on appropriate ordering of the nodes. In \eqref{transformed laplacian}, the set of nodes related to $L_i$  corresponds to a cluster for $i=1,\ldots,\mu$. The remaining clusters formed by the nodes related to $M$ can be identified by investigating algebraic condition of the kernel base of $\bar{L}$ as discussed in \cite{MONACO201953}.

Fig. \ref{fig_cluster} shows that clusters of two networks with the same graph topology but different edge weights, which can be obtained by using the method in \cite{MONACO201953}. In both networks, nodes 6 and 7 are related to $M$ in \eqref{transformed laplacian}. Depending on edge weights, one cluster is formed by nodes 6 and 7 in Fig. \ref{fig_cluster_ex_1} while two clusters are formed in Fig. \ref{fig_cluster_ex_2}.

The method proposed in \cite{MONACO201953} is useful to analyze clustering behavior in networks. However, the method requires exact information about edge weights, which is not realistic in some cases. 
For instance, it is challenging to identify the exact values of edge weights in a social network. Further, edge weights of a network can be time-varying. This observation suggests a necessity for an weight-independent notion of clusters, which is introduced in the following subsection.


\subsection{Topological clusters}  \label{sec_topological_clusters}

In this section, we introduce a new concept called topological clusters, which are determined by purely topological connectivity, independent of the edge weight values. These clusters only depend on topological characteristics of networks.
Fig. \ref{topological_cluster_examples} shows examples of topological clusters. In these networks the nodes in each cluster converge to an identical value independent of edge weights. 

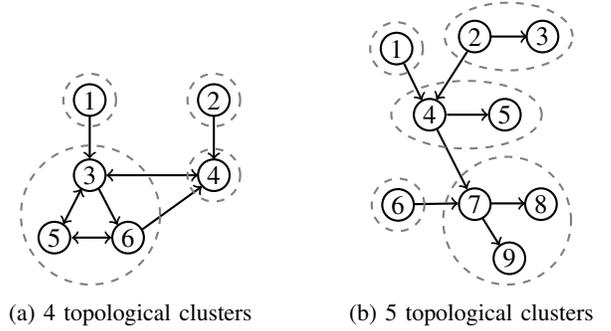
\begin{figure} 
\centering
\begin{subfigure}[b]{0.45\linewidth}
\centering
\begin{tikzpicture}
\draw[thick] (-0.45,2) circle (0.21cm);
\node at (-0.45,2) {1};
\draw[thick] (-0.45,1) circle (0.21cm);
\node at (-0.45,1) {3};
\draw[thick] (-0.92,0.17) circle (0.21cm);
\node at (-0.92,0.17) {5};
\draw[thick] (1.2,1) circle (0.21cm);
\node at (1.2,1) {4};
\draw[thick] (1.2,2) circle (0.21cm);
\node at (1.2,2) {2};
\draw[thick] (0.06,0.17) circle (0.21cm);
\node at (0.06,0.17) {6};
\draw[thick,->] (-0.45,1.79) -- (-0.45,1.22);
\draw[thick,->] (1.2,1.79) -- (1.2,1.21);
\draw[thick,->] (0.25,0.28) -- (1.04,0.85);
\draw[thick,<->] (-0.16,0.17) -- (-0.69,0.17);
\draw[thick,<->] (-0.56,0.82) -- (-0.8,0.35);
\draw[thick,->] (-0.33,0.82) -- (-0.08,0.34);
\draw[thick,<->] (-0.23,1) -- (0.99,1);
\draw[gray,thick,dashed] (1.2,2) circle (0.35cm);
\draw[gray,thick,dashed] (-0.45,2) circle (0.35cm);
\draw[gray,thick,dashed] (1.2,1) circle (0.35cm);
\draw[rotate=0,gray,thick,dashed] (-0.44,0.47) ellipse (0.925cm and 0.925cm);
\end{tikzpicture}
\caption{4 topological clusters}
\end{subfigure}
\hspace{3mm}
\begin{subfigure}[b]{0.45\linewidth}
\centering
\begin{tikzpicture}
\draw[thick] (-0.2,0.04) circle (0.21cm);
\node at (-0.2,0.04) {5};
\draw[thick] (-1.2,0.04) circle (0.21cm);
\node at (-1.2,0.04) {4};
\draw[thick] (-0.6,-1.15) circle (0.21cm);
\node at (-0.6,-1.15) {7};
\draw[thick] (-1.64,0.92) circle (0.21cm);
\node at (-1.64,0.92) {1};
\draw[thick] (0.3,1.08) circle (0.21cm);
\node at (0.3,1.08) {3};
\draw[thick] (0.28,-1.15) circle (0.21cm);
\node at (0.28,-1.15) {8};
\draw[thick] (-0.14,-1.87) circle (0.21cm);
\node at (-0.14,-1.87) {9};
\draw[thick] (-0.6,1.08) circle (0.21cm);
\node at (-0.6,1.08) {2};
\draw[thick] (-1.62,-1.16) circle (0.21cm);
\node at (-1.62,-1.16) {6};

\draw[thick,->] (-1.54,0.72) -- (-1.3,0.23);
\draw[thick,->] (-0.7,0.88) -- (-1.11,0.24);
\draw[thick,<-] (-0.42,0.04) -- (-0.97,0.04);
\draw[thick,<-] (-0.7,-0.95) -- (-1.1,-0.16);
\draw[thick,->] (-0.5,-1.33) -- (-0.27,-1.7);
\draw[thick,->] (-0.4,1.08) -- (0.1,1.08);
\draw[thick,<-] (-0.82,-1.14) -- (-1.4,-1.15);
\draw[thick,->] (-0.4,-1.14) -- (0.07,-1.14);
\draw[gray,thick,dashed] (-1.64,0.92) circle (0.35cm);
\draw[gray,thick,dashed] (-1.62,-1.16) circle (0.35cm);
\draw[rotate=0,gray,thick,dashed] (-0.15,1.08) ellipse (0.85cm and 0.45cm);
\draw[rotate=0,gray,thick,dashed] (-0.71,0.04) ellipse (1.0cm and 0.45cm);
\draw[rotate=0,gray,thick,dashed] (-0.17,-1.37) ellipse (0.85cm and 0.85cm);
\end{tikzpicture}
\caption{5 topological clusters}
\end{subfigure}
\caption{Examples of topological clusters.}
\label{topological_cluster_examples}
\end{figure}

To clearly characterize clustering properties due to topological characteristics, we define topological clusters as follows:
\begin{definition} [Topological cluster]  \label{def_topological_cluster} 
For a network, a topological cluster is defined as a maximal set of nodes that converge to an identical value independent of edge weights.  
\end{definition}


In the above definition, a topological cluster is a maximal set because it cannot be expanded by addition of any node. That is, if a node is added to a topological cluster, the nodes of the expanded set may not converge to an identical value. In contrast, a topological cluster can be a part of a cluster. 

Since topological clusters are dependent only on node dynamics and interaction topology, they can be found by analyzing the interaction topology once node dynamics are given. The following theorem provides a necessary and sufficient condition for topological clusters of \eqref{consensus_protocol_2}:
\begin{theorem} \label{thm_topological_clusters}
For a  network \eqref{consensus_protocol_2} over a digraph $\mathcal{G} = (\mathcal{V}, \mathcal{E}, A)$, a subset $\mathcal{S} \subset \mathcal{V}$ is a topological cluster if and only if $\mathcal{S}$ is a maximal set subject to the following conditions:
\begin{enumerate}
    \item[(C1)] The induced subgraph $\mathcal{G}[\mathcal{S}]$ has at least one  spanning tree;
    \item[(C2)] In $\mathcal{S}$, only one root node of $\mathcal{G}[\mathcal{S}]$ or no node has incoming edges from $\mathcal{V}\backslash\mathcal{S}$.
\end{enumerate}
\end{theorem}

\begin{proof}
{\color{black}
(Sufficiency) To prove sufficiency, we consider the following two cases: the set $\mathcal{S}$ is a maximal set such that $\mathcal{G}[\mathcal{S}]$ has at least one spanning tree, and either %
\begin{itemize}
    \item no node in $\mathcal{S}$ has incoming edges from $\mathcal{V} \backslash \mathcal{S}$ or
    \item only one root node of $\mathcal{G}[\mathcal{S}]$ has incoming edges from $\mathcal{V} \backslash \mathcal{S}$.
\end{itemize}
In the following, we show that the nodes in $\mathcal{G[S]}$ reach consensus for each case.

We first assume that no node in $\mathcal{S}$ has incoming edges from $\mathcal{V}\backslash\mathcal{S}$. Denote by $x_{\mathcal{S}}$ the concatenated vector of the nodes in $\mathcal{S}$. The node dynamics can be written as follows:
\begin{align*}
    \dot{x}_{\mathcal{S}} = -L_{\mathcal{S}} x_{\mathcal{S}},
\end{align*}
where $L_{\mathcal{S}}$ is the Laplacian matrix of the induced subgraph $\mathcal{G}[\mathcal{S}]$. According to \textbf{Lemma \ref{lemma_Laplacian_1}} and \textbf{\ref{lemma_Laplacian_2}}, since $\mathcal{G}[\mathcal{S}]$ has at least one spanning tree, the nodes in $\mathcal{S}$ reach consensus regardless of the edge weights.

We then consider the case where only one root node of $\mathcal{G[S]}$ has incoming edges from $\mathcal{V} \backslash \mathcal{S}$. Let $x_{\mathcal{S}}$ and $x_{\mathcal{V}\backslash\mathcal{S}}$ denote the concatenated vector of the nodes in $\mathcal{S}$ and $\mathcal{V}\backslash\mathcal{S}$, respectively. Then $\dot{x} = -Lx$ can be partitioned into 
\begin{align} \label{eq_partitioned_dynamics}
    \left[ \begin{array}{l}
	\dot{x}_{\mathcal{S}}\\
	\dot{x}_{\mathcal{V}\backslash\mathcal{S}}\\
\end{array} \right] = - \left[ \begin{matrix}
	L_{\mathcal{S}}&		L_{\mathcal{S}-\mathcal{V}\backslash\mathcal{S}}\\
	L_{\mathcal{V}\backslash\mathcal{S}-\mathcal{S}} &		L_{\mathcal{V}\backslash\mathcal{S}}\\
\end{matrix} \right] \left[ \begin{array}{l}
	x_{\mathcal{S}} \\
	x_{\mathcal{V}\backslash\mathcal{S}}\\
\end{array} \right].
\end{align}

Since $\mathcal{G}[\mathcal{S}]$ has at least one  spanning tree, $L_{\mathcal{S}}$ is invertible \cite{liu2012necessary}. Further it follows that $x_{\mathcal{S}}$ and $x_{\mathcal{V}\backslash\mathcal{S}}$ converge to constant vectors as $t \to \infty$ \cite{cao2009containment}. Thus we obtain
\begin{align*}
   \mathbf{0} =  L_{\mathcal{S}} x^{ss}_{\mathcal{S}} + L_{\mathcal{S}-\mathcal{V}\backslash\mathcal{S}} x^{ss}_{\mathcal{V}\backslash\mathcal{S}},
\end{align*}
where $x^{ss}_{\mathcal{S}}$ and $x^{ss}_{\mathcal{V}\backslash\mathcal{S}}$ denote the steady-state constant vector of the nodes in $\mathcal{S}$ and $\mathcal{V}\backslash\mathcal{S}$, respectively. Due to the invertibility of $L_{\mathcal{S}}$, it is obvious that $x^{ss}_{\mathcal{S}}$ is uniquely determined for given $x^{ss}_{\mathcal{V}\backslash\mathcal{S}}$. Suppose that $r \in \mathcal{S}$ is the root node having incoming edges from $\mathcal{V} \backslash \mathcal{S}$. Then $\dot{x}_{\mathcal{S}} = -L_{\mathcal{S}} x_{\mathcal{S}} - L_{\mathcal{S}-\mathcal{V}\backslash\mathcal{S}} x_{\mathcal{V}\backslash\mathcal{S}}$ can be written as
\begin{align*}
    \dot{x}_{i} &= \sum_{j \in \mathcal{S}} a_{ij}(x_{j} - x_{i}), \; i \in \mathcal{S}, i \not = r, \\
    \dot{x}_{r} &= \sum_{j \in \mathcal{S}} a_{rj}(x_{j} - x_{r}) + \sum_{j \in \mathcal{V}\backslash\mathcal{S}} a_{rj}(x_{j} - x_{r}),
\end{align*}
which results in
\begin{align*}
    x_{\mathcal{S}}^{ss} = \mathbf{1}_{|\mathcal{S}|} \frac{\sum_{j \in \mathcal{V}\backslash\mathcal{S}} a_{rj}x^{ss}_{j}}{\sum_{j \in \mathcal{V}\backslash\mathcal{S}} a_{rj}}.
\end{align*}
}
Hence, nodes belonging to $\mathcal{S}$ converge to an identical steady-state value $x_{\mathcal{S}}^{ss}$ whatever edge weights are.

(Necessity) The necessity could be proved by a contraposition. Let us consider the following three cases:
\begin{itemize}
    \item $\mathcal{G}[\mathcal{S}]$ does not have a  spanning tree;
    \item Any nodes that are not root nodes, or more than or equal to two root nodes of $\mathcal{G}[\mathcal{S}]$ have incoming edges from $\mathcal{V}\backslash\mathcal{S}$;
    \item $\mathcal{S}$ is not maximal while conditions in (C1) and (C2) are satisfied.
\end{itemize}
In the following, we show that $\mathcal{S}$ is not a topological cluster for each of the above three cases. 

First, suppose that $\mathcal{G}[\mathcal{S}]$ does not have a  spanning tree. It follows from \textbf{Lemma \ref{lemma_partitioning_G}} that $\mathcal{S}$ can be partitioned into $\mathcal{S}_{1}$, $\mathcal{S}_{2}$, and $\mathcal{S}_{3}$ such that
\begin{itemize}
    \item $\mathcal{S}_{1}$ and $\mathcal{S}_{2}$ are nonempty and $\mathcal{G}[\mathcal{S}_{1}]$ is strongly connected;
    \item no nodes in $\mathcal{S}_{1}$ and $\mathcal{S}_{2}$ have incoming edges from $\mathcal{S}\backslash\mathcal{S}_{1}$ and $\mathcal{S}\backslash\mathcal{S}_{2}$, respectively;
    \item  $\mathcal{S}_{3}$ corresponds to $\mathcal{S}\backslash(\mathcal{S}_{1}\cup \mathcal{S}_{2})$.
\end{itemize}
%
Based on the partitioning, $L_{\mathcal{S}}$ in \eqref{eq_partitioned_dynamics} can be written as follows:
\begin{align*}
    L_{\mathcal{S}} = \left[ \begin{matrix}
	L_{\mathcal{S}_{1}}&		\mathbf{0}&		\mathbf{0}\\
	\mathbf{0}&		L_{\mathcal{S}_{2}}&		\mathbf{0}\\
L_{\mathcal{S}_{31}}&		L_{\mathcal{S}_{32}}&		L_{\mathcal{S}_{3}}\\
\end{matrix} \right].
\end{align*}
Then it can be shown that $\mathcal{S}_{1}$ and $\mathcal{S}_{2}$ converge to different values depending on edge weights of the network as follows:
\begin{itemize}
    \item If $\mathcal{S}_{1}$ and $\mathcal{S}_{2}$ do not have incoming edges from $\mathcal{V}\backslash\mathcal{S}$, the steady-state values of $\mathcal{S}_{1}$ and $\mathcal{S}_{2}$ are determined independently of each other. Thus $\mathcal{S}$ is not a topological cluster.
    \item If only one of $\mathcal{S}_{1}$ and $\mathcal{S}_{2}$ has incoming edges from $\mathcal{V}\backslash\mathcal{S}$, the steady-state value of the nodes in the set is determined independently of the nodes in the other set, which means that $\mathcal{S}$ is not a topological cluster.
    \item Suppose that both of $\mathcal{S}_{1}$ and $\mathcal{S}_{2}$ have incoming edges from $\mathcal{V}\backslash\mathcal{S}$. 
    \begin{itemize}
        \item Suppose that $\mathcal{S}_{1}$ has only a single incoming edge from a node in $\mathcal{V}\backslash\mathcal{S}$. Since $\mathcal{S}_{1}$ is strongly connected and has no incoming edges from $\mathcal{S}\backslash\mathcal{S}_{1}$, $\mathcal{S}_{1}$ belongs to another topological cluster that includes the starting node of the incoming edge. Thus $\mathcal{S}$ is not a topological cluster. 
        \item If $\mathcal{S}_{1}$ has more than one incoming edges only from the nodes of a topological cluster in $\mathcal{V}\backslash\mathcal{S}$, the steady-state value of $\mathcal{S}_{1}$ is equal to that of the topological cluster in $\mathcal{V}\backslash\mathcal{S}$. This means that $\mathcal{S}_{1}$ belongs to the topological cluster in $\mathcal{V}\backslash\mathcal{S}$. Thus $\mathcal{S}$ is not a topological cluster.
        \item If $\mathcal{S}_{1}$ has more than one incoming edges from the nodes of multiple topological clusters in $\mathcal{V}\backslash\mathcal{S}$ and $\mathcal{S}_{2}$ also has incoming edges, the steady-state values of $\mathcal{S}_{1}$ and $\mathcal{S}_{2}$ are dependent on the weights for the incoming edges, which means that $\mathcal{S}$ is not a topological cluster.
    \end{itemize}
\end{itemize}

Second, suppose that any nodes that are not root nodes of $\mathcal{G}[\mathcal{S}]$, or more than or equal to two root nodes of $\mathcal{G}[\mathcal{S}]$ have incoming edges from $\mathcal{V}\backslash\mathcal{S}$. We consider only the case where $\mathcal{G}[\mathcal{S}]$ has at least one  spanning tree because $\mathcal{S}$ cannot be a topological cluster in the absence of  spanning trees as shown above.
%
%
%

Denote the set of the root nodes of $\mathcal{G}[\mathcal{S}]$ by $\mathcal{S}_{r}$ and the set of remaining nodes by $\mathcal{S}_{f}$. Then $\dot{x} = -Lx$ can be written as
\begin{align*} 
    \left[ \begin{array}{l}
	\dot{x}_{\mathcal{S}_{f}}\\
	\dot{x}_{\mathcal{S}_{r}}\\
	\dot{x}_{\mathcal{V}\backslash\mathcal{S}}\\
\end{array} \right] = - \left[ \begin{matrix}
	L_{\mathcal{S}_f}&		L_{\mathcal{S}_{f}-\mathcal{S}_{r}}&		L_{\mathcal{S}_{f}-\mathcal{V}\backslash\mathcal{S}}\\
	\mathbf{0}&		L_{\mathcal{S}_r}&		L_{\mathcal{S}_{r}-\mathcal{V}\backslash\mathcal{S}}\\
	L_{\mathcal{V}\backslash\mathcal{S}-\mathcal{S}_{r}}&		L_{\mathcal{V}\backslash\mathcal{S}-\mathcal{S}_{f}}&		L_{\mathcal{V}\backslash\mathcal{S}}\\
\end{matrix} \right] \left[ \begin{array}{l}
	x_{\mathcal{S}_{f}}         \\
	x_{\mathcal{S}_{r}}         \\
	x_{\mathcal{V}\backslash\mathcal{S}} \\
\end{array} \right].
\end{align*}
Since $\mathcal{G}[\mathcal{S}]$ has at least one  spanning tree, $L_{\mathcal{S}}$ is invertible; thus $L_{\mathcal{S}_{f}}$ is also invertible \cite[Lemma 4]{liu2012necessary}. Then $x_{\mathcal{S}_{f}}^{ss}$ is uniquely determined as follows: 
\begin{align*}
   x_{\mathcal{S}_{f}}^{ss} =  - L_{\mathcal{S}_f}^{-1}L_{\mathcal{S}_{f}-\mathcal{S}_{r}}x_{\mathcal{S}_{r}}^{ss} -   L_{\mathcal{S}_f}^{-1}L_{\mathcal{S}_{f}-\mathcal{V}\backslash\mathcal{S}} x_{\mathcal{V}\backslash\mathcal{S}}^{ss}.
\end{align*}
Consider the case where only some nodes that are not root nodes of $\mathcal{G}[\mathcal{S}]$ have incoming edges from $\mathcal{V}\backslash\mathcal{S}$. The dynamics of the nodes in $\mathcal{S}_{f}$ and $\mathcal{S}_{r}$ can be written as
\begin{align*}
    \dot{x}_{i} &= \sum_{j \in \mathcal{S}_{f}} a_{ij}(x_{j} - x_{i}) + \sum_{j \in \mathcal{S}_{r}} a_{ij}(x_{j} - x_{i}) \\
    & \quad + \sum_{j \in \mathcal{V}\backslash\mathcal{S}} a_{ij}(x_{j} - x_{i}), \; i \in \mathcal{S}_{f}, \\
    \dot{x}_{i} &= \sum_{j \in \mathcal{S}_{r}} a_{ij}(x_{j} - x_{i}), \; i \in \mathcal{S}_{r},
\end{align*}
Since the nodes in $\mathcal{S}_{r}$ do not have any incoming edges from $\mathcal{V}\backslash\mathcal{S}_{r}$, they converge to the steady-state value that is dependent on $\mathcal{G}[\mathcal{S}_{r}]$ and the initial values. However, the steady-state values of the nodes in $\mathcal{S}_{f}$ are dependent on $x_{\mathcal{V}\backslash\mathcal{S}}$. Therefore, the steady-state values of the nodes in $\mathcal{S}_{r}$ and $\mathcal{S}_{f}$ cannot be identical for arbitrary edge weights and initial values, which means that $\mathcal{S}$ is not a topological cluster.

We next consider the case where more than or equal to two nodes have incoming edges from $\mathcal{V}\backslash\mathcal{S}$ and at least one of them is a root node of $\mathcal{G}[\mathcal{S}]$. We exclude the case where all incoming edges into the nodes of $\mathcal{S}$ come out from the same topological cluster because it makes the nodes of $\mathcal{S}$ belonging to the other topological cluster. Let $x_{r}$ be the root node which has incoming edges from $\mathcal{V}\backslash\mathcal{S}$, and $x_k$ be another node which has incoming edges from $\mathcal{V}\backslash\mathcal{S}$. Then we have
\begin{align*}
    \dot{x}_{r} &= \sum_{j \in \mathcal{S}_{r}} a_{rj}(x_{j} - x_{r}) + \sum_{j \in \mathcal{V}\backslash\mathcal{S}} a_{rj}(x_{j} - x_{r}), \\
    \dot{x}_{k} &= \sum_{j \in \mathcal{S}} a_{kj}(x_{j} - x_{k}) + \sum_{j \in \mathcal{V}\backslash\mathcal{S}} a_{kj}(x_{j} - x_{k}).
\end{align*}
Suppose that $\mathcal{S}$ is a topological cluster. Then, we have
\begin{align*}
   x_{r}^{ss} &= \frac{\sum_{j \in \mathcal{V}\backslash\mathcal{S}} a_{rj}x_{j}^{ss}}{\sum_{j \in \mathcal{S}_{r}} a_{rj} + \sum_{j \in \mathcal{V}\backslash\mathcal{S}} a_{rj}} \\
    & =   \frac{\sum_{j \in \mathcal{V}\backslash\mathcal{S}} a_{kj}x_{j}^{ss}}{\sum_{j \in \mathcal{S}} a_{kj} + \sum_{j \in \mathcal{V}\backslash\mathcal{S}} a_{kj}} = x_{k}^{ss},
\end{align*}
which cannot be identically true for arbitrary edge weights. This contradiction shows that $\mathcal{S}$ is not a topological cluster.

Finally, suppose that $\mathcal{S}$ is not maximal while conditions in (C1) and (C2) are satisfied. It follows from the proof for sufficiency that all the nodes of $\mathcal{S}$ converge to an identical value if conditions in (C1) and (C2) are satisfied. Since $\mathcal{S}$ is not maximal, one can expand the set $\mathcal{S}$ by adding at least one node in $\mathcal{V}\backslash\mathcal{S}$ while satisfying conditions in (C1) and (C2), which shows that $\mathcal{S}$ is not a topological cluster.

\end{proof}


\begin{remark}
Several remarks are provided:
\begin{itemize}
    \item Though we focus on \eqref{consensus_protocol_2}, the result developed in what follows can be applied to a general network model studied in \cite{wieland2011internal}.
    \item If a topological cluster in the network \eqref{consensus_protocol_2} has incoming edges from the other topological clusters, the edges need to come out from two or more topological clusters. If not, the maximality condition in \textbf{Theorem \ref{thm_topological_clusters}} cannot be satisfied.
    \item There may exist topological clusters consisting of only a single node. For instance, a topological cluster can consist of only one leader node. Further there may exist topological clusters consisting of only single node that has incoming edges from two or more topological clusters. An existence of topological clusters of a single node implies that every node belongs to a topological cluster.
    \item The root node of a topological cluster $\mathcal{S}$ is defined as the root node of a  spanning tree in the induced subgraph $\mathcal{G}[\mathcal{S}]$. 
    \item If a subset $\mathcal{S}$ satisfies (C1) and (C2) in \textbf{Theorem \ref{thm_topological_clusters}} but does not satisfy the maximality condition, all nodes in $\mathcal{S}$ belong to the same topological cluster, but $\mathcal{S}$ is not a topological cluster.
\end{itemize}
\label{remark_topological_cluster}
\end{remark}

Based on the proof of \textbf{Theorem \ref{thm_topological_clusters}}, topological clusters can be categorized into two classes as follows:
\begin{definition} [Leader and follower topological clusters]
A topological cluster is a leader topological cluster if it has no incoming edges coming out from the other topological clusters. On the other hand it is called a follower topological cluster if it is not a leader topological cluster.
\label{def_leader_follower_topological_cluster}
\end{definition}

The convergence value of a leader topological cluster is determined independently of the other topological clusters. In this sense, a leader topological cluster can be regarded as a leader in the network.
In contrast, a follower topological cluster has multiple incoming edges staring from two or more topological clusters. As shown in the proof of \textbf{Theorem \ref{thm_topological_clusters}}, the  convergence value of a follower topological cluster is determined by the convergence values of the topological clusters from which the incoming edges come out.

\begin{figure} 
\centering
\begin{tikzpicture}
\node at (-2.42,2.08) {1};
\draw[thick] (-2.42,2.08) circle (0.21cm);
\node at (-2.91,1.23) {2};
\draw[thick] (-2.91,1.23) circle (0.21cm);
\node at (-1.88,1.23) {3};
\draw[thick] (-1.88,1.23) circle (0.21cm);
\node at (-1.25,0.69) {4};
\draw[thick] (-1.25,0.69) circle (0.21cm);
\node at (-0.28,2.31) {5};
\draw[thick] (-0.28,2.31) circle (0.21cm);
\node at (0.62,2.31) {6};
\draw[thick] (0.62,2.31) circle (0.21cm);
\node at (-0.28,1.4) {7};
\draw[thick] (-0.28,1.4) circle (0.21cm);
\node at (0.62,1.4) {8};
\draw[thick] (0.62,1.4) circle (0.21cm);
\node at (-3.3,-0.01) {9};
\draw[thick] (-3.3,-0.01) circle (0.21cm);
\node at (-2.5,-0.69) {10};
\draw[thick] (-2.5,-0.69) circle (0.21cm);
\node at (-1.53,-0.76) {11};
\draw[thick] (-1.53,-0.76) circle (0.21cm);
\node at (-2.08,-1.59) {12};
\draw[thick] (-2.08,-1.59) circle (0.21cm);
\node at (-0.4,0.12) {13};
\draw[thick] (-0.4,0.12) circle (0.21cm);
\node at (-0.2,-0.86) {14};
\draw[thick] (-0.2,-0.86) circle (0.21cm);
\node at (0.54,-0.26) {15};
\draw[thick] (0.54,-0.26) circle (0.21cm);
\node at (-0.68,-1.89) {16};
\draw[thick] (-0.68,-1.89) circle (0.21cm);
\draw[thick,->] (-2.52,1.88) -- (-2.8,1.42);             
\draw[thick,->] (-2.32,1.88) -- (-2,1.41);             
\draw[thick,->] (-2.69,1.23) -- (-2.1,1.23);                 
\draw[thick,->] (-1.72,1.09) -- (-1.42,0.83);           
\draw[thick,->] (-1.95,1.02) -- (-2.41,-0.49);          
\draw[thick,->] (-1.07,0.57) -- (-0.58,0.24);           
\draw[thick,<->] (-0.06,2.31) -- (0.4,2.31);              
\draw[thick,->] (-0.28,2.09) -- (-0.28,1.61);               
\draw[thick,->] (-0.08,1.4) -- (0.4,1.4);               
\draw[thick,->] (-0.31,1.19) -- (-0.37,0.33);            
\draw[thick,->] (-3.14,-0.15) -- (-2.67,-0.56);          
\draw[thick,<->] (-2.28,-0.71) -- (-1.75,-0.75);          
\draw[thick,->] (-2.4,-0.89) -- (-2.17,-1.39);             
\draw[thick,->] (-1.65,-0.94) -- (-1.96,-1.41);             
\draw[thick,->] (-1.86,-1.64) -- (-0.9,-1.85);         
\draw[thick,->] (-0.35,-0.1) -- (-0.23,-0.65);          
\draw[thick,->] (-0.19,0.05) -- (0.33,-0.18);        
\draw[thick,->] (-0.03,-0.74) -- (0.37,-0.4);          
\draw[thick,->] (-0.29,-1.05) -- (-0.58,-1.7);             
\node at (-3.35,0.59) {C};
\draw[gray,thick,dashed] (-3.3,-0.01) circle (0.35cm);
\node at (-3.02,-1.7) {D};
\draw[gray,thick,dashed] (-0.68,-1.89) circle (0.35cm);
\node at (-0.09,-1.99) {F};
\draw[gray,thick,dashed] (-0.04,-0.27) circle (0.95cm);
\node at (-0.68,2.73) {B};
\draw[gray,thick,dashed] (0.17,1.87) circle (1.cm);
\node at (0.96,-0.86) {E};
\draw[gray,thick,dashed] (-2.07,-0.99) circle (0.975cm);
\node at (-1.66,2.44) {A};
\draw[rotate=-38,gray,thick,dashed] (-2.4557,-0.2957) ellipse (1.4cm and 1.cm);
\end{tikzpicture}
\caption{Illustration of topological clusters} 
\label{examples_various_topological_cluster}
\end{figure}

We provide an example to illustrate the notion of topological clusters. Fig. \ref{examples_various_topological_cluster} illustrates various types of topological clusters marked as $A,B,C,D,E$ and $F$ . First, $A$, $B$, and $C$ are leader topological clusters. A leader topological cluster may or may not have a leader node. For instance, $A$ has a leader node while $B$ has not. A leader topological cluster may consist of single node as $C$. Second, $D$, $E$, and $F$ are follower topological clusters. As discussed in \textbf{Theorem \ref{thm_topological_clusters}} and \textbf{Remark \ref{remark_topological_cluster}}, it is shown that the induced subgraph of each follower topological cluster has at least one root node and the corresponding  spanning tree. Accordingly, we can predict the consensus value of each topological cluster as follows:
\begin{itemize}
    \item Cluster $A$: the initial value of node 1.
    \item Cluster $B$: the weighted sum of the initial values of nodes 5 and 6.
    \item Cluster $C$: the initial value of node 9.
    \item Cluster $D$: the weighted sum of the initial values of nodes $A$ and $C$.
    \item Cluster $E$: the weighted sum of the initial values of nodes $A$ and $B$.
    \item Cluster $F$: the weighted sum of the initial values of nodes $D$ and $E$.
\end{itemize}

\section{Topological Clustering Algorithm}  \label{sec_clustering_algorithm}


Prior to discussing the topological clustering algorithm, let us introduce leader-like strongly connected components (LSCCs). In Fig. \ref{examples_various_topological_cluster}, nodes 5 and 6 are strongly connected and have no incoming edges from other parts of the graph, except from each other. These nodes always belong to the same topological cluster and behave like leader nodes. Based on this observation, we define LSCC as follows:
\begin{definition}  [Leader-like strongly connected component]
A leader-like strongly connected component (LSCC) of a graph is a strongly connected component that has no incoming edges from the other part of the graph.
\end{definition}

We next define LSCC condensation. Since all nodes in an LSCC always converge to an identical value, it is convenient to contract the LSCC into a node when considering topological clusters. For example, the LSCC condensation of Fig. \ref{examples_various_topological_cluster} can be observed in Fig. \ref{fig_LSCC_example}. The nodes 5 and 6 in Fig. \ref{examples_various_topological_cluster} are contracted onto node 5 in Fig. \ref{fig_LSCC_example}. Based on this contraction, LSCC condensation is defined as follows:
\begin{definition} [LSCC condensation]
The LSCC condensation of a graph $\mathcal{G}$ is a graph obtained by contracting each LSCC in $\mathcal{G}$ onto a single node.  
\end{definition}

\begin{figure} 
\centering
\begin{tikzpicture}
\node at (-2.42,2.08) {1};
\draw[thick] (-2.42,2.08) circle (0.21cm);
\node at (-2.91,1.23) {2};
\draw[thick] (-2.91,1.23) circle (0.21cm);
\node at (-1.88,1.23) {3};
\draw[thick] (-1.88,1.23) circle (0.21cm);
\node at (-1.25,0.69) {4};
\draw[thick] (-1.25,0.69) circle (0.21cm);
\node at (0.17,2.31) {5};
\draw[thick] (0.17,2.31) circle (0.21cm);
\node at (-0.28,1.4) {6};                               
\draw[thick] (-0.28,1.4) circle (0.21cm);
\node at (0.62,1.4) {7};                               
\draw[thick] (0.62,1.4) circle (0.21cm);
\node at (-3.3,-0.01) {8};                               
\draw[thick] (-3.3,-0.01) circle (0.21cm);
\node at (-2.5,-0.69) {9};                               
\draw[thick] (-2.5,-0.69) circle (0.21cm);
\node at (-1.53,-0.76) {10};                               
\draw[thick] (-1.53,-0.76) circle (0.21cm);
\node at (-2.08,-1.59) {11};                               
\draw[thick] (-2.08,-1.59) circle (0.21cm);
\node at (-0.4,0.12) {12};                               
\draw[thick] (-0.4,0.12) circle (0.21cm);
\node at (-0.2,-0.86) {13};                               
\draw[thick] (-0.2,-0.86) circle (0.21cm);
\node at (0.54,-0.26) {14};                               
\draw[thick] (0.54,-0.26) circle (0.21cm);
\node at (-0.68,-1.89) {15};                               
\draw[thick] (-0.68,-1.89) circle (0.21cm);
\draw[thick,->] (-2.52,1.88) -- (-2.8,1.42);             
\draw[thick,->] (-2.32,1.88) -- (-2,1.41);             
\draw[thick,->] (-2.69,1.23) -- (-2.1,1.23);                 
\draw[thick,->] (-1.72,1.09) -- (-1.42,0.83);           
\draw[thick,->] (-1.95,1.02) -- (-2.41,-0.49);          
\draw[thick,->] (-1.07,0.57) -- (-0.58,0.24);           
\draw[thick,->] (0.1,2.09) -- (-0.13,1.58);               
\draw[thick,->] (-0.08,1.4) -- (0.4,1.4);               
\draw[thick,->] (-0.31,1.19) -- (-0.37,0.33);            
\draw[thick,->] (-3.14,-0.15) -- (-2.67,-0.56);          
\draw[thick,<->] (-2.28,-0.71) -- (-1.75,-0.75);          
\draw[thick,->] (-2.4,-0.89) -- (-2.17,-1.39);             
\draw[thick,->] (-1.65,-0.94) -- (-1.96,-1.41);             
\draw[thick,->] (-1.86,-1.64) -- (-0.9,-1.85);         
\draw[thick,->] (-0.35,-0.1) -- (-0.23,-0.65);          
\draw[thick,->] (-0.19,0.05) -- (0.33,-0.18);        
\draw[thick,->] (-0.03,-0.74) -- (0.37,-0.4);          
\draw[thick,->] (-0.29,-1.05) -- (-0.58,-1.7);             
\node at (-3.35,0.59) {C};
\draw[gray,thick,dashed] (-3.3,-0.01) circle (0.35cm);
\node at (-3.02,-1.7) {D};
\draw[gray,thick,dashed] (-0.68,-1.89) circle (0.35cm);
\node at (-0.09,-1.99) {F};
\draw[gray,thick,dashed] (-0.04,-0.27) circle (0.95cm);
\node at (-0.68,2.73) {B};
\draw[gray,thick,dashed] (0.17,1.87) circle (1.cm);
\node at (0.96,-0.86) {E};
\draw[gray,thick,dashed] (-2.07,-0.99) circle (0.975cm);
\node at (-1.66,2.44) {A};
\draw[rotate=-38,gray,thick,dashed] (-2.4557,-0.2957) ellipse (1.4cm and 1.cm);
\end{tikzpicture}
\caption{LSCC condensation of network in Fig. \ref{examples_various_topological_cluster}} 
\label{fig_LSCC_example}
\end{figure}
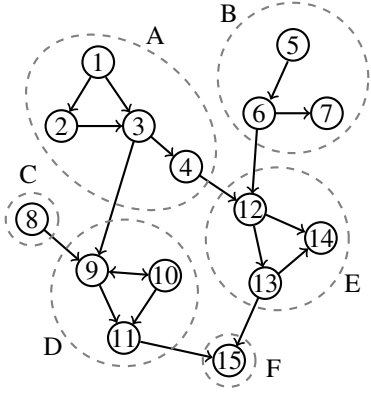

Due to the convergence property of nodes in each LSCC, it is obvious that the topological clusters of \eqref{consensus_protocol_2} over $\mathcal{G}$ are the same as those of the network over the LSCC condensation of $\mathcal{G}$. Accordingly, in this paper, it is assumed that there is no LSCCs in \eqref{consensus_protocol_2} for simplicity. Now, we are ready to define the cluster root and cluster follower nodes as follows:
\begin{definition} [Cluster root and cluster follower nodes] \label{def_CR_CF}
For \eqref{consensus_protocol_2}, a node is a cluster root (CR) node if it is a leader node or the root node that has incoming edges from the outside of its topological cluster. A node is a cluster follower (CF) node if it is not a CR node.
\end{definition}

Considering \textbf{Definition \ref{def_CR_CF}}, it follows from  the second condition of \textbf{Theorem \ref{thm_topological_clusters}} that each topological cluster has exactly one CR node. In Fig. \ref{fig_LSCC_example}, nodes 1, 5, 8, 9, 12, and 15 are CR nodes, and there is one in each cluster. Furthermore, all other nodes are CF nodes. The identification whether it is CR or CF node is crucial for the topological clustering algorithm. The following lemma shows connectivity between CR and CF nodes:

\begin{lemma}
For  network \eqref{consensus_protocol_2}, every CF node has at least one  path from the CR node in its topological cluster.
\label{lem_relation_CR_CF}
\end{lemma}

\begin{proof}
Consider an arbitrary CF node. From \textbf{Definition \ref{def_CR_CF}}, the CF node is neither a leader node nor the root node that has incoming edges from the outside of its topological cluster. It follows from \textbf{Theorem \ref{thm_topological_clusters}} and \textbf{Definition \ref{def_CR_CF}} that the CF node lies on a  spanning tree of the induced subgraph of a topological cluster. Further the root node of the  spanning tree is the CR node of the topological cluster. Therefore, the CF node has at least one  path from the CR node in its topological cluster.
\end{proof}


According to the following theorem, a node with one or zero incoming edges can be classified as a CR or CF node.

\begin{theorem} \label{thm_node_classification_1}
For network \eqref{consensus_protocol_2} over the LSCC condensation of a digraph, the following is true:
\begin{itemize}
\item A node without an incoming edge is a CR node.
\item A node with only one incoming edge is a CF node.
\end{itemize}
\end{theorem}

\begin{proof}
First, a node without an incoming edge is a leader node. It follows from \textbf{Definition \ref{def_CR_CF}} that the node is a CR node.

Second, let $i$ be a node with only one incoming edge. If $j$ is a neighbor of $i$, it follows from the maximality condition in \textbf{Theorem \ref{thm_topological_clusters}} that both $i$ and $j$ belong to the same topological cluster. Suppose that $i$ is the CR node of the topological cluster. Since $i$ has no incoming edge from the outside of the topological cluster, it is a leader node, which is a contradiction. Therefore, a node with only one incoming node is a CF node.
\end{proof}

For example, nodes 1, 5, and 8 in Fig. \ref{fig_LSCC_example} have no incoming edges and they are CR nodes. On the other hand, nodes 2, 4, 6, 7, 10, and 13 have only one incoming edge and they are classified as CF nodes. However, it can be quite challenging to classify nodes with two or more incoming edges as either CR or CF nodes. To this end, we first define popular nodes as follows:
\begin{definition} [Popular node] \label{def_popular_node}
For  network \eqref{consensus_protocol_2}, a node with two or more incoming edges is called a popular node.
\end{definition}

According to \textbf{Definition \ref{def_popular_node}}, nodes 3, 9, 11, 12, 14, and 15 in Fig. \ref{fig_LSCC_example} are popular nodes. Next, to propose a theorem classifying these popular nodes, the following lemma is necessary:
\begin{lemma}
For a graph $\mathcal{G}=(\mathcal{V},\mathcal{E})$ that has no LSCC, every node that is not a leader node always has at least one path from a leader node.
\label{lem_popular_node}
\end{lemma}
\begin{proof}
Consider an arbitrary node $i \in \mathcal{V}$, which is not a leader node. Let $\mathcal{S} \subset \mathcal{V}$ be the set of all nodes that have at least one path to the node $i$. Accordingly, $\mathcal{S}$ has no incoming edge from the outside. We prove the lemma by contradiction. Suppose that every node in $\mathcal{S}$ is not a leader node, which means that every node in $\mathcal{S}$ has at least one incoming edge from the other nodes. Obviously, the out-degree of each nodes in $\mathcal{S}$ is more than or equal to 1. Hence, a node in $\mathcal{S}$ forms at least one directed cycle together with some other nodes, which means that each node in $\mathcal{S}$ belongs to a strongly connected component. Denote the strongly connected components by $\mathcal{S}_1, \mathcal{S}_2, ..., \mathcal{S}_N$. Due to the assumption that there is no LSCC in $\mathcal{G}$, each strongly connected component has at least one incoming edge from one of the other strongly connected components. Applying the same argument repeatedly, it can be concluded that some of the strongly connected components form an LSCC, which is a contradiction.
Therefore, there exists at least one leader node in $\mathcal{S}$, which implies that node $i$ has at least one path from a leader node.

\end{proof}


Based on \textbf{Lemma \ref{lem_popular_node}}, we propose a theorem that classifies the popular nodes into CR and CF nodes.

\begin{theorem} \label{thm_common_node_finder}
For  network \eqref{consensus_protocol_2}, a popular node is a CF node if and only if all acyclic paths from any leader node to the popular node contain at least one common node other than the popular node.
\end{theorem}
\begin{proof}
(Sufficiency) Let $i$ be a popular node and assume that all acyclic paths from a leader node to $i$ share at least one common node other than $i$. Let $j$ be the farthest node from $i$ among the common nodes. Denote by $\mathcal{N}$ the set of nodes contained in all acyclic paths from $j$ to $i$. If there exist any incoming edges from the outside of $\mathcal{N}$ to nodes in $\mathcal{N}$ other than $j$, $j$ cannot be the farthest common node, which is a contradiction. This means that only $j$ has incoming edges from the outside of $\mathcal{N}$. Further $j$ has at least one directed path to every other node in $\mathcal{N}$. That is, $\mathcal{N}$ satisfies the conditions (C1) and (C2) of \textbf{Theorem \ref{thm_topological_clusters}}. In general, $\mathcal{N}$ does not satisfy the maximality condition of \textbf{Theorem \ref{thm_topological_clusters}}. Yet, as discussed in  \textbf{Remark \ref{remark_topological_cluster}}, all nodes in $\mathcal{N}$ belong to the same topological cluster. Assume that node $k$ is the CR node of the topological cluster. It follows from \textbf{Theorem \ref{thm_topological_clusters}} that there is a directed path from $k$ to every other node in $\mathcal{N}$. Since every incoming edge from outside of the topological cluster can be connected to only $k$, $k$ is the farthest common node of all acyclic paths from a leader node to $i$, which implies that $k$ corresponds to $j$. Therefore, $j$ is a CR node of the topological cluster and thus $i$ is a CF node.

(Necessity) Let $i$ be a popular node and assume that it is a CF node. Let $j$ be the CR node of the topological cluster to which $i$ belongs. It follows from \textbf{Theorem \ref{thm_topological_clusters}} that there exists at least one directed path from $j$ to every other node in the topological cluster. Further only $j$ has incoming edges starting from outside of the topological cluster. This implies that $j$ is a common node of all acyclic paths from a leader node to $i$, which completes the proof.
\end{proof}

According to the \textbf{Theorem \ref{thm_common_node_finder}}, nodes 3, 11, and 14 in Fig. \ref{fig_LSCC_example} are CF nodes. For instance, node 11 has many acyclic paths from any leader node, and node 9 is a common node in all of those paths. On the other hand, node 9, 12, and 15 in Fig. \ref{fig_LSCC_example} are CR nodes because they have no common node in all of those paths.

Based on the above discussions, we present the following topological clustering algorithm, which takes the graph of network \eqref{consensus_protocol_2} and yields information about topological clusters:  
\begin{enumerate}
    \item Obtain LSCC condensation for given graph. 
    \item Classify the nodes that have no incoming edges or exactly one incoming edge as CR or CF node, respectively.
    \item Find all acyclic paths from a leader node to each popular node.
    \begin{enumerate}
        \item If there exists a common node that the paths share other than the popular node, classify the popular node as a CF node.
        \item Otherwise, classify the popular node as a CR node.
    \end{enumerate}
    \item Identify all topological clusters based on graph connectivity and the classification found in the above steps.
\end{enumerate}

In the above topological clustering algorithm, the LSCC condensation of $\mathcal{G}$ can be obtained based on the existing algorithms that search strongly connected components \cite{cormen2009introduction}. By checking the connectivity of each strongly connected component with the other part of the graph, LSCC condensation can be obtained. The pseudo-code of the topological clustering algorithm is provided in \textbf{Algorithm \ref{alg_classification}} and \textbf{Algorithm \ref{alg_searching}}. It is assumed that the graph $\mathcal{G}$ has no LSCCs for simplicity. In \textbf{Algorithm \ref{alg_classification}}, each node is classified into a CR or CF node. The set of each CR node and its corresponding CF nodes is then identified as a topological cluster. This requires searching all acyclic paths from a leader node to every popular node as implemented in \textbf{Algorithm \ref{alg_searching}}. Note that \textbf{Algorithm \ref{alg_searching}} is called in line 15 of \textbf{Algorithm \ref{alg_classification}}.

\begin{algorithm}
	\caption{Search topological clusters}
	\begin{algorithmic}[1]
        \State{\textbf{Input}: graph $\mathcal{G}=(\mathcal{V},\mathcal{E})$}
        \State Initialize $\mathcal{P}$, the set of popular nodes
        \State Initialize $\mathcal{L}$, the set of leader nodes
		\For {each $v$ in $\mathcal{V}$} 
            \If {$v$ has no incoming edge}
                \State Classify $v$ as a CR node, and add $v$ to $\mathcal{L}$
            \ElsIf {$v$ has only one incoming edge}
                \State Classify $v$ as a CF node
            \ElsIf {$v$ has more than one incoming edges}
                \State Add $v$ to $\mathcal{P}$
            \EndIf
		\EndFor
		\For {each node $v$ in $\mathcal{P}$}
            \State Obtain $\mathcal{M}_{v}$ by calling \textbf{Algorithm \ref{alg_searching}} with $v$, $\mathcal{G}$, $\mathcal{L}$, and $\mathcal{D}_{v} = \{v\}$
            \If {all rows of $\mathcal{M}_{v}$ share a common node other than $v$}
                \State Classify $v$ as a CF node
            \ElsIf {all rows of $\mathcal{M}_{v}$ do not share a common node other than $v$}
                \State Classify $v$ as a CR node
            \EndIf
		\EndFor
    \State Identify the set of each CR node and its corresponding CF nodes as a topological cluster
    \State{\textbf{Output}: topological clusters}
	\end{algorithmic} 
    \label{alg_classification}
\end{algorithm} 

\textbf{Algorithm \ref{alg_searching}} takes a popular node $v$ as input and identifies all acyclic paths that start from a leader node and end at $v$. It returns a matrix with rows containing the information of the acyclic paths.

\begin{algorithm}
	\caption{Search acyclic paths} 
	\begin{algorithmic}[1]
	    \State{\textbf{Input}: popular node $v$, graph $\mathcal{G}$, set of leader nodes $\mathcal{L}$, and a set $D_{v}$ for checking duplication of nodes}
	    \State Initialize $\mathcal{M}_v$, matrix with paths from leader nodes to  $v$ in rows
		\For {each $e$ in $\mathcal{E}$}
		    \State {Find a node $n$ such that $e=(v,n) \in \mathcal{E}$}
		    \If {$n$ $\notin$ $\mathcal{D}_{v}$}
		        \If{$n \in \mathcal{L}$}
                    \State Append a path $v,n$ to an empty row of $\mathcal{M}_v$
		        \ElsIf {$n \notin \mathcal{L}$}
		            \State Obtain $\mathcal{M}_{n}$ by calling \textbf{Algorithm \ref{alg_searching}} with $n$, $\mathcal{G}$, $\mathcal{L}$, and $\mathcal{D}_{n} = \mathcal{D}_{v} \cup \{n\}$
		            \State Insert $v$ to the first of every row of $\mathcal{M}_n$
		            \State Append each row of $\mathcal{M}_n$ to the empty rows of $\mathcal{M}_v$
		        \EndIf
		    \EndIf
		\EndFor
	    \State{\textbf{Output}: $\mathcal{M}_v$}
	\end{algorithmic} 
    \label{alg_searching}
\end{algorithm}

\section{Example}    \label{sec_examples}
We now provide examples that illustrate the analysis and use of the topological clustering algorithm. 


\subsection{Comparison with state trajectory}

In the first example, we identify topological clusters in the network shown in \textbf{Fig.} \ref{fig_LSCC_example}, which is the LSCC condensation of the network shown in \textbf{Fig.} \ref{examples_various_topological_cluster}.
According to \textbf{Algorithm \ref{alg_classification}} presented in \textbf{Section} \ref{sec_clustering_algorithm}, the set of leader nodes is $\mathcal{L} = \{1,5,8\}$. Further, $\{2,4,6,7,10,13\}\subset \mathcal{N}^{CF}$, where $\mathcal{N}^{CF}$ is the set of CF nodes. Subsequently, the set of popular nodes is given as $\mathcal{P} = \{3,9,11,12,14,15\}$. By \textbf{Algorithm \ref{alg_searching}} and \textbf{Theorem \ref{thm_common_node_finder}}, the popular nodes are identified as CR or CF nodes: $\{9,12,15\}\subset \mathcal{N}^{CR}$, $\{3,11,14\}\subset \mathcal{N}^{CF}$, where $\mathcal{N}^{CR}$ is the set of CR nodes. The identification result is as follows:
\begin{align*}
\mathcal{N}^{CR} &= \{1,5,8,9,12,15\}, \\
\mathcal{N}^{CF} &= \{2,3,4,6,7,10,11,13,14\}
\end{align*}
By matching each CR node with the appropriate CF nodes, the topological clusters are identified. The clustering result is summarized in \textbf{Table} \ref{table_Classification_result_ex1}.

\begin{table}
\caption{Clustering result of example A}
\centering
\begin{tabular}{c|c|p{2.5cm}}\toprule \hline
\renewcommand{\tabcolsep}{0.7mm}	
Cluster no. &  CR node &  CF node \\ \hline\hline
1		&	1		& 2, 3, 4  \\ \hline
2		&	5		& 6, 7  \\ \hline
3		&	8		& - \\ \hline
4		&	9		& 10, 11  \\ \hline
5		&	12		& 13, 14 \\ \hline
6		&	15		& - \\ \hline
\bottomrule 
\end{tabular}
\label{table_Classification_result_ex1}
\end{table}


\begin{figure}
\centering
\begin{subfigure}{1\linewidth}
\centering
\includegraphics[width=0.95\linewidth]{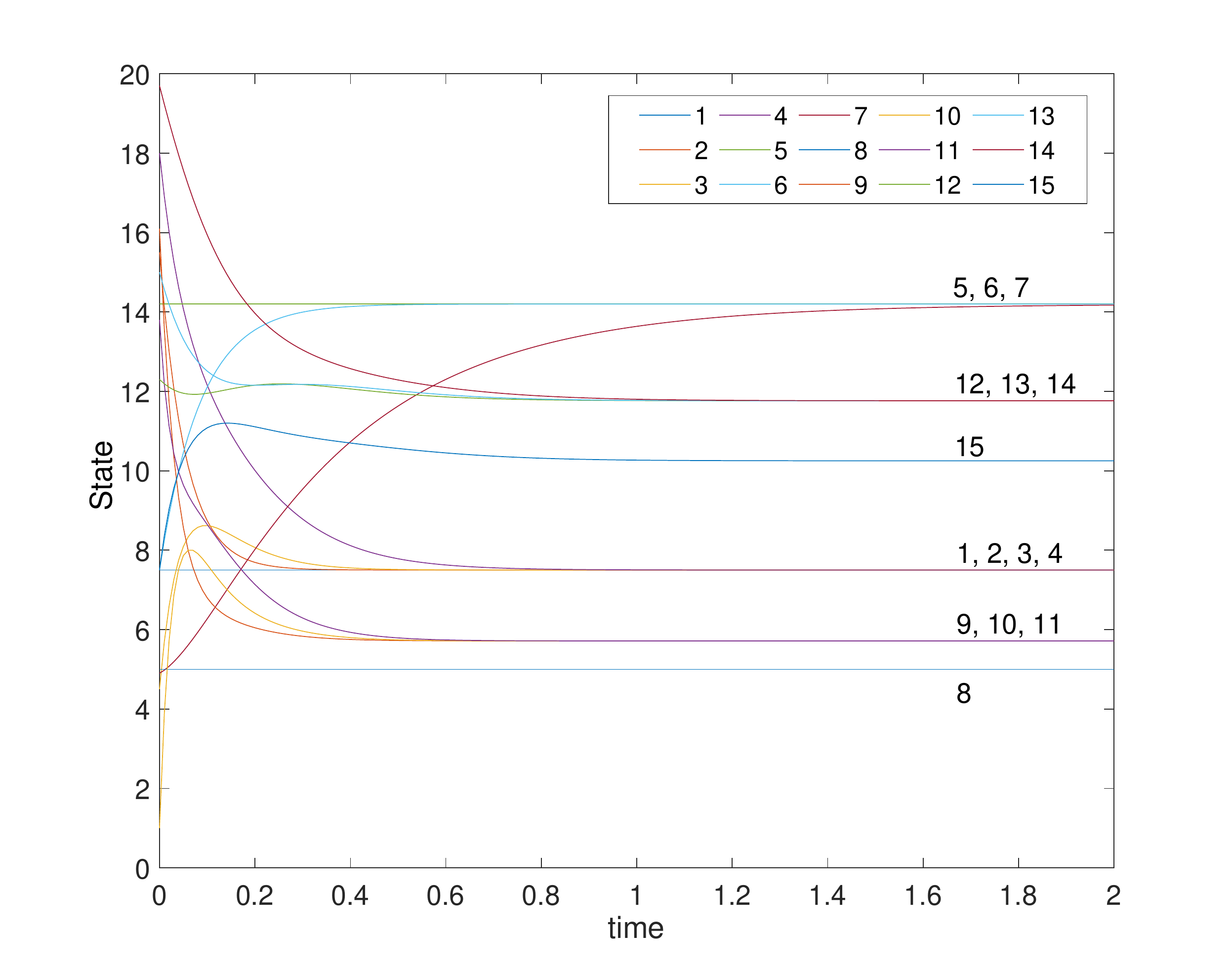}
\caption{State trajectory with random edge weight set 1}
\label{fig_ex1_trajectory1}
\end{subfigure}
\vspace{1mm}
\begin{subfigure}{1\linewidth}
\centering
\includegraphics[width=0.95\linewidth]{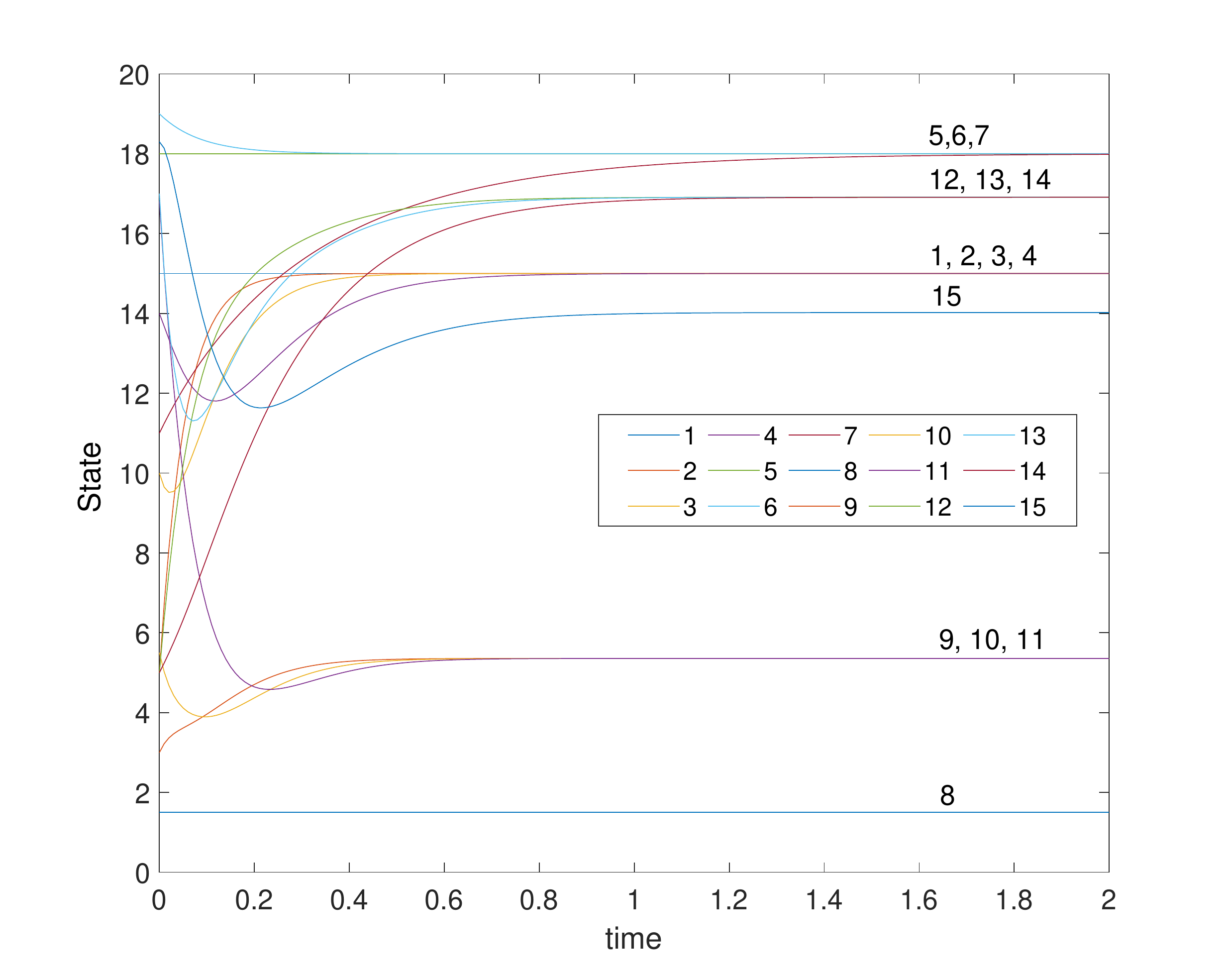}
\caption{State trajectory with random edge weight set 2}
\label{fig_ex1_trajectory2}
\end{subfigure}
\caption{State trajectories of network in example A}
\label{fig_ex1_trajectory}
\end{figure}

To validate the result in \textbf{Table} \ref{table_Classification_result_ex1}, we investigate the state trajectory of the network by simulation. For the network shown in \textbf{Fig.} \ref{fig_LSCC_example}, we randomly selected initial state values of nodes between 1 and 20, and edge weights were chosen as random positive numbers less than 5. \textbf{Fig.} \ref{fig_ex1_trajectory} illustrates the state trajectories for two cases. In \textbf{Fig.} \ref{fig_ex1_trajectory1} and \textbf{Fig.} \ref{fig_ex1_trajectory2}, we can observe that the nodes belong to a topological cluster converge to an identical value for both cases.

\subsection{Comparison with existing works}

\begin{figure} 
\centering
\begin{tikzpicture}
\draw[thick] (-1.2,1) circle (0.21cm);          
\node at (-1.2,1) {1};
\draw[thick] (0,1) circle (0.21cm);             
\node at (0,1) {2};
\draw[thick] (-0.6,0) circle (0.21cm);          
\node at (-0.6,0) {3};
\draw[thick] (2,1) circle (0.21cm);           
\node at (2,1) {4};
\draw[thick] (1.2,0.3) circle (0.21cm);           
\node at (1.2,0.3) {5};
\draw[thick] (2.7,-0.1) circle (0.21cm);        
\node at (2.7,-0.1) {6};
\draw[thick] (-1.2,-1.9) circle (0.21cm);       
\node at (-1.2,-1.9) {7};
\draw[thick] (0.6,-1.7) circle (0.21cm);        
\node at (0.6,-1.7) {8};
\draw[thick] (2.3,-1.9) circle (0.21cm);        
\node at (2.28,-1.9) {9};
\draw[thick,<->] (-0.99,1) -- (-0.21,1);        
\draw[thick,->] (-1.1,0.83) -- (-0.7,0.17);     
\draw[thick,<->] (1.84,0.87) -- (1.36,0.43);    
\draw[thick,->] (-0.65,-0.21) -- (-1.15,-1.7);    
\draw[thick,->] (-0.495,-0.18) -- (0.495,-1.52);      
\draw[thick,->] (-0.41,-0.12) -- (2.14,-1.765);   
\draw[thick,->] (1.05,0.15) -- (-1.05,-1.75);        
\draw[thick,->] (1.13,0.1) -- (0.67,-1.5);        
\draw[thick,->] (1.27,0.1) -- (2.23,-1.7);        
\draw[thick,->] (2.5,-0.17) -- (-0.98,-1.85);        
\draw[thick,->] (2.54,-0.235) -- (0.76,-1.565);        
\draw[thick,->] (2.66,-0.306) -- (2.34,-1.694);        
\end{tikzpicture}
\caption{Network topology in example B} 
\label{examples_topological_cluster_2}
\end{figure}
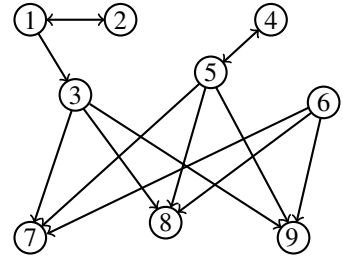

In the second example, we compare a topological cluster and a cluster obtained based on the method proposed in \cite{MONACO201953}. Consider a network in \textbf{Fig.} \ref{examples_topological_cluster_2} with a weight matrix $A = [a_{ij}] \in \mathbb{R}^{N \times N}$, where $\forall i,j \in \mathcal{V}, a_{ij} \geq 0$. To clearly illustrate the difference between the two concepts, we randomly choose the edge weights, but make sure they satisfy the following condition:
\begin{align*} 
a_{7,3}=a_{8,3}=a_{9,3}=2\\
a_{7,5}=a_{8,5}=a_{9,5}=3\\
a_{7,6}=a_{8,6}=a_{9,6}=5
\end{align*}

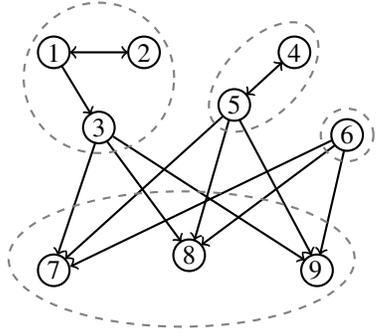
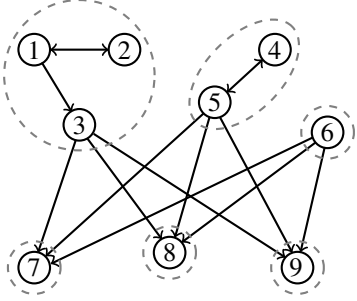
\begin{figure}
\centering
\begin{subfigure}[b]{0.6\linewidth}
\centering
\begin{tikzpicture}
\draw[thick] (-1.2,1) circle (0.21cm);          
\node at (-1.2,1) {1};
\draw[thick] (0,1) circle (0.21cm);             
\node at (0,1) {2};
\draw[thick] (-0.6,0) circle (0.21cm);          
\node at (-0.6,0) {3};
\draw[thick] (2,1) circle (0.21cm);           
\node at (2,1) {4};
\draw[thick] (1.2,0.3) circle (0.21cm);           
\node at (1.2,0.3) {5};
\draw[thick] (2.7,-0.1) circle (0.21cm);        
\node at (2.7,-0.1) {6};
\draw[thick] (-1.2,-1.9) circle (0.21cm);       
\node at (-1.2,-1.9) {7};
\draw[thick] (0.6,-1.7) circle (0.21cm);        
\node at (0.6,-1.7) {8};
\draw[thick] (2.3,-1.9) circle (0.21cm);        
\node at (2.28,-1.9) {9};
\draw[thick,<->] (-0.99,1) -- (-0.21,1);        
\draw[thick,->] (-1.1,0.83) -- (-0.7,0.17);     
\draw[thick,<->] (1.84,0.87) -- (1.36,0.43);    
\draw[thick,->] (-0.65,-0.21) -- (-1.15,-1.7);    
\draw[thick,->] (-0.495,-0.18) -- (0.495,-1.52);      
\draw[thick,->] (-0.41,-0.12) -- (2.14,-1.765);   
\draw[thick,->] (1.05,0.15) -- (-1.05,-1.75);        
\draw[thick,->] (1.13,0.1) -- (0.67,-1.5);        
\draw[thick,->] (1.27,0.1) -- (2.23,-1.7);        
\draw[thick,->] (2.5,-0.17) -- (-0.98,-1.85);        
\draw[thick,->] (2.54,-0.235) -- (0.76,-1.565);        
\draw[thick,->] (2.66,-0.306) -- (2.34,-1.694);        
\draw[gray,thick,dashed] (2.7,-0.1) circle (0.35cm);
\draw[gray,thick,dashed] (-0.6,0.66) circle (1cm);
\draw[rotate=45,gray,thick,dashed] (1.6,-0.65) ellipse (0.9cm and 0.5cm);
\draw[gray,thick,dashed] (0.5,-1.75) ellipse (2.3cm and 0.9cm);
\end{tikzpicture}
\caption{Clusters identified by method in \cite{MONACO201953}}
\label{fig_compare_AEP}
\end{subfigure}
\begin{subfigure}[b]{0.6\linewidth}
\centering
\begin{tikzpicture}
\draw[thick] (-1.2,1) circle (0.21cm);          
\node at (-1.2,1) {1};
\draw[thick] (0,1) circle (0.21cm);             
\node at (0,1) {2};
\draw[thick] (-0.6,0) circle (0.21cm);          
\node at (-0.6,0) {3};
\draw[thick] (2,1) circle (0.21cm);           
\node at (2,1) {4};
\draw[thick] (1.2,0.3) circle (0.21cm);           
\node at (1.2,0.3) {5};
\draw[thick] (2.7,-0.1) circle (0.21cm);        
\node at (2.7,-0.1) {6};
\draw[thick] (-1.2,-1.9) circle (0.21cm);       
\node at (-1.2,-1.9) {7};
\draw[thick] (0.6,-1.7) circle (0.21cm);        
\node at (0.6,-1.7) {8};
\draw[thick] (2.3,-1.9) circle (0.21cm);        
\node at (2.28,-1.9) {9};
\draw[thick,<->] (-0.99,1) -- (-0.21,1);        
\draw[thick,->] (-1.1,0.83) -- (-0.7,0.17);     
\draw[thick,<->] (1.84,0.87) -- (1.36,0.43);    
\draw[thick,->] (-0.65,-0.21) -- (-1.15,-1.7);    
\draw[thick,->] (-0.495,-0.18) -- (0.495,-1.52);      
\draw[thick,->] (-0.41,-0.12) -- (2.14,-1.765);   
\draw[thick,->] (1.05,0.15) -- (-1.05,-1.75);        
\draw[thick,->] (1.13,0.1) -- (0.67,-1.5);        
\draw[thick,->] (1.27,0.1) -- (2.23,-1.7);        
\draw[thick,->] (2.5,-0.17) -- (-0.98,-1.85);        
\draw[thick,->] (2.54,-0.235) -- (0.76,-1.565);        
\draw[thick,->] (2.66,-0.306) -- (2.34,-1.694);        
\draw[gray,thick,dashed] (2.7,-0.1) circle (0.35cm);
\draw[gray,thick,dashed] (0.6,-1.7) circle (0.35cm);
\draw[gray,thick,dashed] (2.3,-1.9) circle (0.35cm);
\draw[gray,thick,dashed] (-1.2,-1.9) circle (0.35cm);
\draw[gray,thick,dashed] (-0.6,0.66) circle (1cm);
\draw[rotate=45,gray,thick,dashed] (1.6,-0.65) ellipse (0.9cm and 0.5cm);
\end{tikzpicture}
\caption{Topological clusters}
\label{fig_compare_TC}
\end{subfigure}
\caption{Clusters of network in example B}
\end{figure}

Using the method proposed in \cite{MONACO201953}, we can obtain the result shown in \textbf{Fig.} \ref{fig_compare_AEP}, which is dependent on the edge weights. We can obtain the topological clusters by applying the algorithm proposed in Section \ref{sec_clustering_algorithm}. Different from the result in \textbf{Fig.} \ref{fig_compare_AEP}, \textbf{Fig.} \ref{fig_compare_TC} shows the topological clusters which are independent of edge weights.

\subsection{Real world example}

\begin{figure*}
\centering
\includegraphics[width=2\columnwidth]{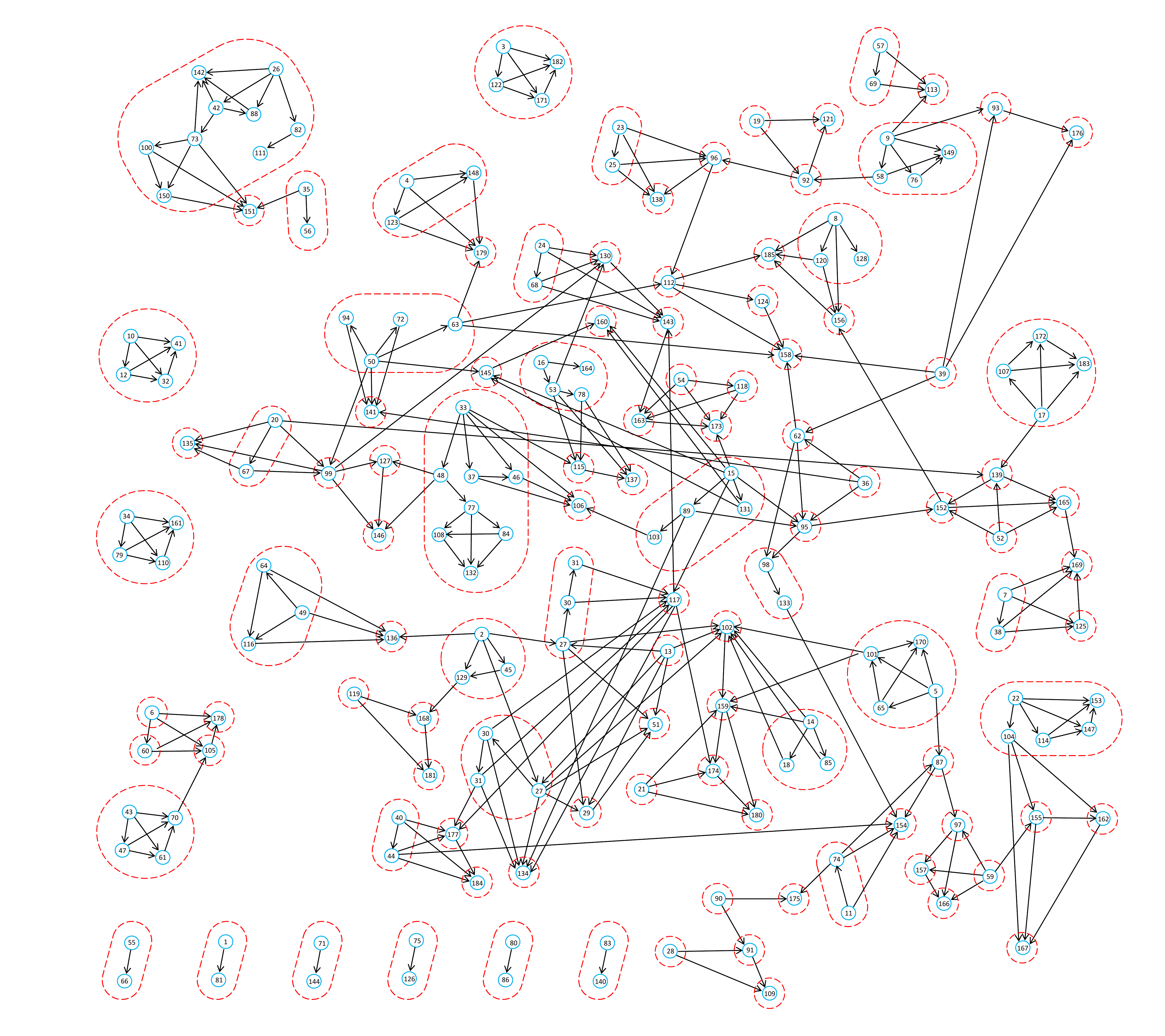}
\caption{Topological clusters of social network in example C}
\label{fig_cluster_social_network}
\end{figure*}

\begin{figure}
\centering
\includegraphics[width=1\linewidth]{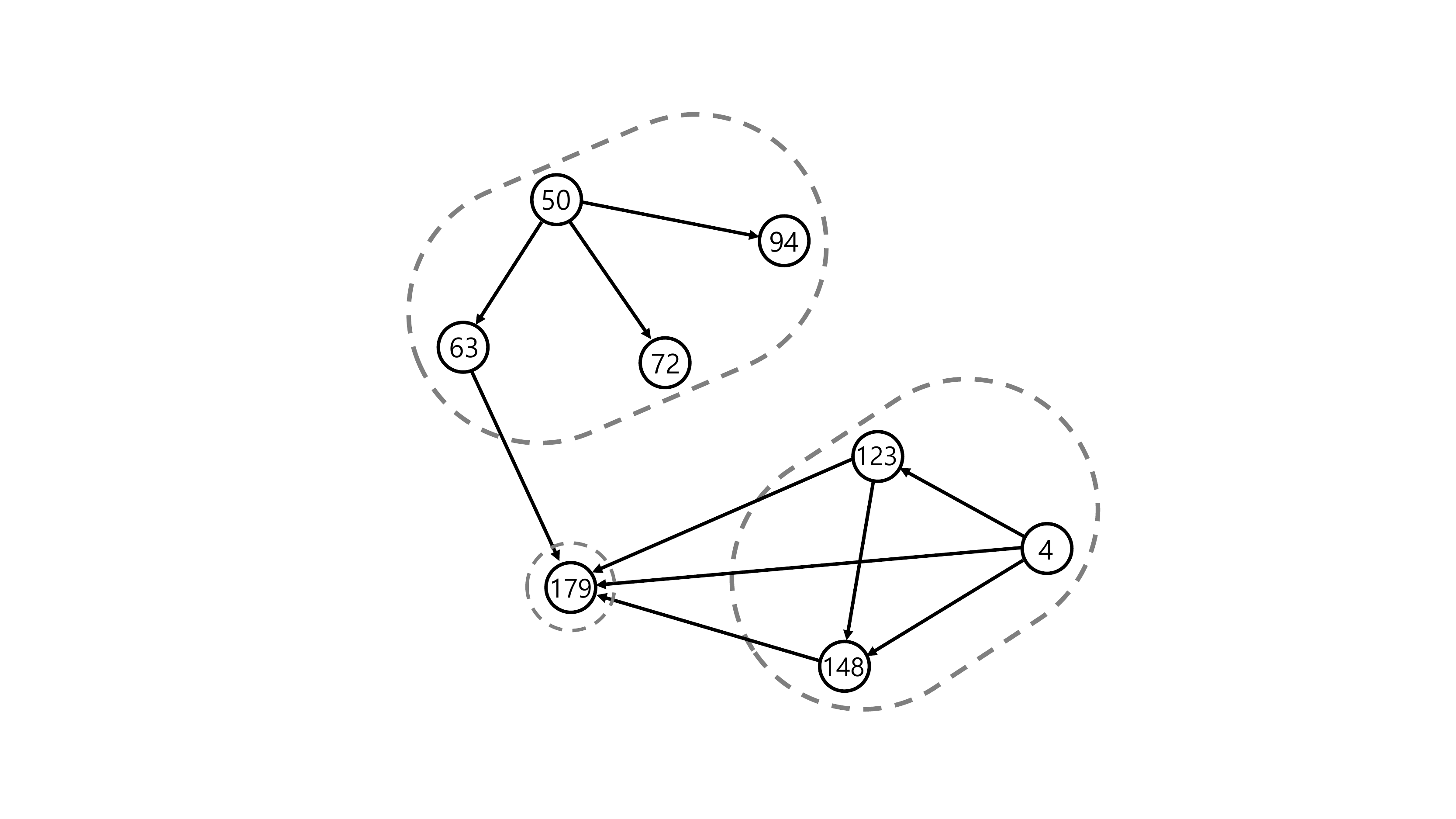}
\caption{Part of social network and topological clusters}
\label{fig_part_cluster_result}
\end{figure}

We apply the topological clustering algorithm to a social network found in \cite{Fire2012student}. The social network, consisting of 185 nodes and 360 edges, shows interaction among students in a course. \textbf{Fig.} \ref{fig_cluster_social_network} shows the results of the topological clustering algorithm for the social network. Solid circles indicate nodes, and dotted circles around nodes indicate topological clusters. The result shown in \textbf{Fig.} \ref{fig_cluster_social_network} clearly identifies the relationship between the topological clusters, which can be useful for analyzing the network.

In \cite{Fire2012student}, the edge weights were set in proportion to the intensity of the students' interactions. The authors of \cite{Fire2012student} measured the intensity by counting the number of online communications among the students, which seems to be ambiguous. Further, the intensity can frequently change. The proposed algorithm can be successfully applied to the network because the notion of the topological clusters is not dependent on the values of edge weights.

Additionally, we can use this result to analyze the opinion propagation in the network. If we assume that the opinion dynamics can be modeled by \eqref{consensus_protocol_2}, we can explain how the opinions propagate among the students. For example, a part of the network are shown in \textbf{Fig.} \ref{fig_part_cluster_result}. Nodes 4 and 50 are the CR nodes of the corresponding topological clusters. The opinion of each node dominates that of each cluster. Regardless of the initial values and the changes in edge weights, nodes 63, 72, and 94 converge to the opinion of node 50. Nodes 123 and 148 converge to the opinion of node 4. Further the opinion of node 179 is between opinions of the two topological clusters.




\section{Conclusion}    \label{sec_conclusion}

We studied clustering behavior in a network of single-integrator nodes. We proposed the notion of topological clusters that is independent of edge weights. We also presented a necessary and sufficient condition for topological clusters. We then provided an algorithm to search topological clusters. The algorithm allows us to obtain information about topological clusters from the graph representing the interaction topology. Examples validated the analysis and algorithm.







\end{document}